\documentclass[11pt,eqno]{amsart}
\usepackage{mathrsfs}

\usepackage{graphicx}
\usepackage{amsmath}
\usepackage{amsthm}
\usepackage{amssymb}
\usepackage{amsfonts}
\usepackage{epsfig}
\usepackage{hyperref}
\def\arg {\mathop{\rm arg}\nolimits}

\newtheorem{pro}{Proposition}
\newtheorem{cor}{Corollary}
\newtheorem{lem}{Lemma}
\newtheorem{thm}{Theorem}

\theoremstyle{remark}
\newtheorem{rem}[thm]{Remark}

\numberwithin{equation}{section}

\begin{document}

\title[ PIII$'$  $\tau$-function sequence and random matrix theory]{Applications in random matrix theory of a PIII$'$ $\tau$-function sequence from Okamoto's Hamiltonian formulation}

\author{Dan Dai}
\address{Department of Mathematics, City University of Hong Kong, Tat Chee Avenue, Kowloon, Hong Kong}
\email{dandai@cityu.edu.hk}
\author{Peter J. Forrester}
\address{School of Mathematics and Statistics,
	ARC Centre of Excellence for Mathematical
	and Statistical Frontiers,
	University of Melbourne, Victoria 3010, Australia}
\email{pjforr@unimelb.edu.au}
\author{Shuai-Xia Xu}
\address{Institut Franco-Chinois de l'Energie Nucl\'eaire, Sun Yat-sen University, Guangzhou 510275, China}
\email{xushx3@mail.sysu.edu.cn}

\begin{abstract}
	We consider the singular linear statistic of the Laguerre unitary ensemble consisting of the sum of the reciprocal of the eigenvalues. It is observed that the exponential generating function for this statistic can be written as a
	Toeplitz determinant with entries given in terms of particular
	$K$ Bessel functions. Earlier studies have identified the same determinant, but with the $K$ Bessel functions replaced by $I$ Bessel functions, as relating to the hard edge scaling limit of
	a generalized gap probability for the Laguerre unitary ensemble, in the case of non-negative integer Laguerre parameter. We show that the Toeplitz determinant formed from an arbitrary linear combination of these two Bessel functions occurs as a $\tau$-function sequence in Okamoto's Hamiltonian formulation of Painlev\'e
	III$'$, and consequently the logarithmic derivative of both
	Toeplitz determinants satisfies the same $\sigma$-form Painlev\'e III$'$ differential equation, giving an explanation of a fact which can be observed from earlier results. In addition, some insights into the relationship between this
	characterization of the generating function, and its characterization in the $n \to \infty$ limit, both with the Laguerre parameter $\alpha$ fixed, and with $\alpha = n$ (this latter circumstance being relevant to an application to the distribution of the Wigner time delay statistic), are given.
	\end{abstract}

\maketitle


\section{Introduction}
\subsection{$\sigma$-form of PIII$'$ and an LUE matrix integral}

In random matrix theory, the Laguerre unitary ensemble (LUE) is a probability distribution on the space
of $n \times n$ positive definite matrices, specified by the joint eigenvalue probability density function
\begin{equation}\label{def: LUE-PDF}
p(\lambda_1,\dots,\lambda_n)=\frac{1}{C_n} \prod_{l=1}^n\lambda_l^{\alpha}e^{-\lambda_l}\prod_{1\leqslant j<k\leqslant n}(\lambda_j-\lambda_k)^2, \qquad  \alpha>-1, \  \lambda_1,\dots,\lambda_n \in [0, \infty),
\end{equation}
where
\begin{equation}\label{LUE-Zn}
  C_n = \prod_{j=1}^n j! \, \Gamma(j+\alpha)
\end{equation}
is the normalisation;
see e.g.~ the texts Mehta \cite[Eq.~(17.6.5)]{m} or Forrester \cite[Eq.~(3.16)]{F}.
	This probability distribution is realized by the eigenvalues of matrices of the form $M=X^{*} X$, where $X$ is an $m\times n$ ($m\geq n$) rectangular matrix of i.i.d.~entries with standard complex Gaussian distribution $\mathbf{N}[0,1] + i \mathbf{N}[0,1]$. In this setting, the parameter $\alpha$ in \eqref{def: LUE-PDF} equals $m-n$, and is thus a non-negative integer. The terminology ``Laguerre"
	relates to the appearance of the Laguerre weight function $\lambda^\alpha e^{-\lambda}$, $\lambda > 0$, from orthogonal polynomial
	theory, while ``unitary" refers to the fact that matrices $\{M\}$ as defined above have distribution unchanged by the unitary conjugation $M$
	maps to $U^* M U$.
	
One	of the many applications of the Laguerre unitary ensemble is in the theory of quantum scattering; see e.g.~the review \cite{Be97}.
There, a basic setting is the coupled lead-cavity system, and a class of observables are the proper delay times
$\{ \tau_j\}_{j=1}^n$, defined as the eigenvalues of the Wigner--Smith matrix
\begin{equation}\label{QS}
Q := - i \hbar S^{-1} {\partial S \over \partial E}.
\end{equation}
In (\ref{QS}), $S$ is the scattering matrix and $E$ is the energy of the waves as they enter the cavity.
 The proper delay times
are used to compute the Wigner delay time $\tau_{\rm W} := {1 \over n} \sum_{j=1}^n \tau_j$. It gives a measure of
the average time a wave packet spends scattering in the cavity; see \cite{Te15}. For a particular random matrix model
of the lead-cavity system, and assuming the presence of a time reversal symmetry breaking magnetic field, it
was shown by Brouwer et al.~\cite{BFB} that the reciprocal of the proper delay times multiplied by $n$, i.e.~$\{n/\tau_j\}_{j=1}^n$, have the
joint distribution of (\ref{def: LUE-PDF}) with $\alpha = n$.
	
As made explicitly by Mezzadri and Simm \cite{MS}, the problem of quantifying the statistics of the proper
delay times is therefore a particular case of the problem of quantifying the linear statistic
\begin{equation}\label{L}
L=\sum_{k=1}^n \frac{1}{\lambda_k}
\end{equation}
in the Laguerre unitary ensemble. For this purpose, it is convenient to introduce
the moment generating function for $L$, defined by
\begin{equation}\label{def: generating function}
M_n(t)=\frac{1}{C_n} \int_0^{\infty} \cdots \int_0^{\infty}\prod_{1\leqslant j<k\leqslant n}(\lambda_j-\lambda_k)^2 \prod_{l=1}^n\lambda_l^{\alpha}e^{-\lambda_l-\frac{t}{\lambda_l}} \, d \lambda_l,
\end{equation}
which is also the Laplace transform of the probability density of $L$.
Equivalently, as an average over the space of positive definite matrices belonging to the LUE, we have
\begin{equation}\label{MT}
M_n(t) = \Big \langle e^{- t {\rm Tr} \, M^{-1}} \Big \rangle_{\rm LUE}.
\end{equation}
Mezzadri and Simm \cite{MS} sought a characterization of (\ref{def: generating function}) suited to the
computation of the large $n$ form of the scaled cumulants corresponding to (\ref{L}).

There is an alternative formulation of (\ref{def: generating function}).
Motivated by the weighting of each Lebesgue measure $d \lambda_l$ therein, define
\begin{equation}\label{w0}
w(x;t)=x^{\alpha}e^{-x-\frac{t}{x}}, \qquad x>0,\quad  t>0,
\end{equation}
which is the classical Laguerre weight deformed by a simple pole at the origin in the exponent. Note that since $t>0$, the factor $e^{-\frac{t}{x}}$ is exponentially small when $x \to 0^+$. As a consequence, all the moments
\begin{equation}\label{def: moment}
 \mu_j(t):=\int_0^{\infty}x^jw(x;t) \, dx
\end{equation}
exist for any $\alpha \in \mathbb{R}$. The Hankel determinant associated with \eqref{w0} is
\begin{equation}\label{def: Hankel det}
D_n(t)=\det \biggl[ \mu_{j+k}(t)\biggr ]_{j,k=0}^{n-1}.
\end{equation}
It is well-known that, for a general weight $w(x;t)$,
the Hankel determinant admits the multiple integral representation (see e.g.~\cite[Eq.~(5.75)]{F})
\begin{equation}\label{eq:Hankel-integral}
     D_n(t)=\frac{1}{n!} \int_I \cdots \int_I  \prod_{1\leqslant j<k \leqslant n}(\lambda_j-\lambda_k)^2 \prod_{l=1}^n w(\lambda_l;t) \, d\lambda_l,
  \end{equation}
  where $I$ denotes the support of the weight.
Consequently the moment generating function \eqref{def: generating function} can be expressed in terms of the Hankel determinant 
\begin{equation}\label{MD}
M_n(t)=\frac{D_n(t)}{D_n(0)}.
\end{equation}
Moreover, recalling that the explicit form of the normalization constant $C_n$ for the Laguerre unitary ensemble from \eqref{LUE-Zn}, we have
\begin{equation}
D_n(0)=\frac{C_n}{n!} = \frac{G(n+1)G(n+\alpha+1)}{G(\alpha+1)},
\end{equation}
where $G(z)$ is the Barnes' $G$-function satisfying the functional relation
\begin{equation}
  G(z+1)=\Gamma(z)G(z), \quad \textrm{with } G(1) = 1;
\end{equation}
see \cite[Eq.(5.17.1)]{dlmf}.

A viewpoint in integrable systems theory shows that it is advantageous to
extend the multiple integral (\ref{eq:Hankel-integral}) for general weights $w(\lambda;t)$, depending on at least one
parameter $t$ and vanishing at the
endpoints of $I$, to the family of integrals
\begin{equation}
Z_n[w](\{c_j\}) :=  \int_I \cdots \int_I  \prod_{1\leqslant j<k \leqslant n}(\lambda_j-\lambda_k)^2 \prod_{l=1}^n w(\lambda_l;t) e^{\sum_{j=1}^\infty c_j \lambda_l^j} \, d\lambda_l.
\end{equation}
It is a fundamental property, explained in e.g.~the review \cite{vM} or the text \cite[\S 5.7]{F}, that this family of
integrals, as a function of $\{c_j\}$ satisfy the KP hierarchy of equations, with the simplest nontrivial example being
\begin{equation}\label{4.1}
\bigg ( \Big ( {\partial \over \partial c_1} \Big )^4 + 3 \Big ( {\partial \over \partial c_2} \Big )^2 - 4
 {\partial^2 \over \partial c_1 \partial c_2} \bigg ) \log Z_n +
 6  \Big ( {\partial^2 \over \partial c_1^2} \log Z_n \Big )^2 = 0.
\end{equation}
It is furthermore the case that $Z_n$ satisfies additional differential equations --- referred to as Virasoro constraints ---
which in contrast to the KP hierarchy of equations depend on the explicit form of $w(\lambda;t)$; see again \cite{vM} or \cite[\S 9.10]{F}.
It has been known since the work of Adler and van Moerbeke \cite{AvM} that by combining (\ref{4.1}) and the Virasoro constraints, for
some special weights $w(\lambda;t)$ it is possible to determine a differential equation for $Z_n[w](\{c_j=0\})$ as a function of the parameter $t$.

By following this strategy, Mezzadri and Simm \cite{MS} deduced a characterization of  (\ref{def: generating function}) as the solution of
a particular second order nonlinear equation. Earlier, for the same multiple integral (\ref{def: generating function}), but with a different
interpretation relating to bosonic replica field theories,  Osipov and Kanzieper \cite{OK1} had applied precisely this
strategy and written down the differential equation but without giving the details. They remarked that the obtained equation
relates to Painlev\'e III. In the time between the two works \cite{OK1} and \cite{MS}, Chen and Its \cite{ci} studied
 (\ref{def: generating function}), motivated from its interpretation in random matrix theory as the generating function of
 the linear statistic (\ref{L}), and also by it providing an opportunity for expounding on the ladder operator method from
 the theory of orthogonal polynomials \cite{CIs}, as well as the use of a Riemann-Hilbert analysis for the
 same purpose \cite{KH}. Moreover, they identified the nonlinear equation as a particular example
 of the Jimbo-Miwa-Okamoto $\sigma$-form of Painlev\'e III$'$ (strictly speaking the latter requires some
 interpretation as it is the isomonodromy problem of Painlev\'e III itself, not its variant Painlev\'e III$'$
for which the relationship is established).

 \begin{thm}\label{T1} (Osipov and Kanzieper \cite{OK1}, Chen and Its \cite{ci}, Mezzadri and Simm \cite{MS}.) Specify $M_n(t)$ as the multiple integral (\ref{def: generating function}). We have that the quantity
 	$$
 	y_n(t) = t {d \over dt} \log M_n(t)
 	$$
 	satisfies the second order nonlinear differential equation
 	 \begin{equation}\label{H-equation}
 	(ty_n'')^2=(n-(2n+\alpha)y_n')^2-4(n(n+\alpha)+ty_n'-y_n)y_n'(y_n'-1).
 	\end{equation}
 	\end{thm}

 \begin{rem}\label{R2}
 	We know from \cite{OK} (see subsection \ref{S2.1} below for a brief summary) that for the Hamiltonian formulation of Painlev\'e III$'$ ---
 	which depends on parameters $v_1, v_2$ --- the modified Hamiltonian
 	\begin{equation}\label{hH}
 	h(t) = tH(t) + {1 \over 4} v_1^2
 	-{1 \over 2} t
 	\end{equation}
 	 satisfies the nonlinear differential equation
 	\begin{equation}\label{eq:h}
 	(t h'')^2+(4(h')^2-1)(th'-h)+v_1v_2h'-\frac{1}{4}(v_1^2+v_2^2)=0.
 	\end{equation}
 	Comparing with (\ref{H-equation}) we see that the latter corresponds to setting
 	\begin{equation}\label{hv}
 	(v_1, v_2) = (2n+\alpha, - \alpha)
 	\end{equation}
 	and $y_n(t) = h(t) + {t \over 2} - {\alpha^2 \over 4}$.
 	
 	In (\ref{eq:h}), setting
 	\begin{equation}\label{hH1}
 	h\Big ({t \over 4} \Big )= - \sigma_{III'}(t)  + {t \over 8} + {v_1 v_2 \over 4}
 	\end{equation}
 	shows that
 	$\sigma_{III'}$ satisfies
 	\begin{equation}\label{eq: sigma-equation-1}
 	(t\sigma'')^2-v_1v_2(\sigma')^2+\sigma'(4\sigma'-1)(\sigma-t\sigma')-\frac{(v_1-v_2)^2}{4^3}=0,
 	\end{equation}
 	which following \cite{OK} we take as the standard form of the $\sigma$-Painlev\'e III$'$ equation.
 	
 \end{rem}

 \begin{rem}\label{Re3}
	Integrable systems theory is known to be related to the characterisation of other singular statistics
	in random matrix theory. The case of the weight
	\begin{equation}\label{w1}
	w(x;z,s) = \exp \Big ( - {z^2 \over 2 x^2} + {s \over x} - {x^2 \over 2} \Big ), \qquad z \in \mathbb R \backslash \{0\}, \:
	0 \le s < \infty, \, x \in \mathbb R
	\end{equation}
	was considered by Mezzadri and Mo \cite{MM}, Brightmore et al.~\cite{BMM} and (in the case $s=0$) by Min et~al.~\cite{MLC}.
	Chen and Dai \cite{CD}, and later Chen et al.~\cite{CCF} studied the Hankel determinant
	\eqref{def: Hankel det} for weights
	\begin{equation}\label{w2}
	w(x;t,\alpha,\beta) = x^\alpha (1 - x)^\beta e^{-t/x}, \qquad x \in [0,1], \: \alpha, \beta> 0, \: t \ge 0.
	\end{equation}
	The underlying random matrix ensembles are the Gaussian unitary ensemble for (\ref{w1}) and the Jacobi unitary
	ensemble for (\ref{w2}). The weight
		$$
		w(x;\{t_k\}_{k=1}^{2m}) = \exp \Big ( - n \Big (2 x^2 + \sum_{k=1}^{2m}{t_k \over (x - \lambda)^k}
		\Big ) \Big ), \quad t_{2m} > 0,
		$$
	corresponding to the generating function for a family of singular linear statistics in the Gaussian unitary ensemble
	is 	considered in the recent work by Dai et al.~\cite{DXZ}, while a study of 	\eqref{def: Hankel det} in the case of
	the weights
	$$
	w_k(x;t)=x^{\alpha}e^{-x- ( \frac{t}{x}  )^k}, \qquad x>0,\quad  t>0, \, k \in \mathbb Z^+
	$$
	generalising (\ref{w0}) is given in a work of Atkin et al.~\cite{acm}.
	
	The works \cite{MM}, \cite{BMM},
	\cite{CCF}, \cite{DXZ}, \cite{acm} cited above are actually concerned with an integrable systems characterisation of
	the Hankel determinant in the (scaled) large $n$ limit. For such asymptotic studies of  \eqref{def: Hankel det} with the weight
	(\ref{w0}), see the papers by  Mezzadri and Simm \cite{MS}, Xu et al. \cite{XDZ2015,XDZ}, and also Section \ref{S3} below.
	\end{rem}

\medskip
The analogue of Theorem \ref{T1} in relation to the weight (\ref{w2}) is required for our subsequent working.

\begin{thm}\label{T2} (Chen and Dai \cite{CD}.) Let $D_n(t)$ be given by \eqref{eq:Hankel-integral} with weight
(\ref{w2}), and set $H_n(t) = t {d \over dt} \log D_n(t)$. This latter quantity satisfies the second order nonlinear
differential equation
	\begin{equation}\label{J1}
	(t H_n'')^2 = \Big ( n (n + \alpha + \beta) - H_n + (\alpha + t) H_n'\Big )^2 + 4 H_n' (t H_n' - H_n)(\beta - H_n').
	\end{equation}	
\end{thm}
		
		\begin{rem}
			According to the Hamiltonian formulation of the Painlev\'e V equation due to Okamoto \cite{OK} (see
			\cite[\S 8.2]{F} for a text book treatment), what now is referred to as the Jimbo-Miwa-Okamoto $\sigma$-form
			of  Painlev\'e V is the differential equation
			\begin{equation}\label{PV}
			(t \sigma'')^2 - \Big ( \sigma - t \sigma' + 2 ( \sigma')^2 + (\nu_0 + \nu_1 + \nu_2 + \nu_3) \sigma' \Big )^2 +
			4 \prod_{l=0}^3 ( \nu_l + \sigma') = 0,
			\end{equation}
			where $\{\nu_j\}_{j=0}^3$ are parameters.
			As noted in \cite{CD}, after replacing $H_n$ in Theorem \ref{T2} by $\tilde{H}_n := H_n - n(n+\alpha + \beta)$
			we see by comparing (\ref{J1}) and (\ref{PV}) that $\tilde{H}_n$ satisfies
			the Jimbo-Miwa-Okamoto $\sigma$-form of PV with parameters
			\begin{equation}\label{PVa}
			\nu_0 = 0, \quad \nu_1 = - (n + \alpha + \beta), \quad \nu_2 = n, \quad \nu_3 = - \beta.
			\end{equation}
		\end{rem}
			
\subsection{$\sigma$-form of PIII$'$ and PV and gap probabilities}
Theorems \ref{T1} and \ref{T2} restate known characterisations of the Hankel determinant
\eqref{def: Hankel det}, considered for interpretation within random matrix theory in its equivalent multiple integral form
\eqref{eq:Hankel-integral}, for the weights (\ref{w0})	and (\ref{w2}) as solutions of particular $\sigma$-form
 Painlev\'e III$'$ and  Painlev\'e V equation. Moreover, it has been commented that these Hankel determinants have an
 interpretation as the (exponential) moment generating function	for the linear statistic (\ref{L})
 in the Laguerre and Jacobi unitary ensembles, respectively. One of the main points of the present paper
 is to identify inter-relations between these moment generating functions and other probabilistic quantities
 in random matrix theory, and to detail some consequences.

 Long before the results of Theorems \ref{T1} and \ref{T2}, the $\sigma$-forms of the Painlev\'e equations
 were shown to characterise particular gap probabilities, i.e.~the probability that a certain interval is
 free of eigenvalues, in a variety of random matrix ensembles and their scaling limits. The pioneering work
 along these lines is that of Jimbo et al.~\cite{JMMS}, who expressed the probability of there being no eigenvalues in
 an interval of size $s$ in the bulk scaled circular unitary ensemble, or equivalently bulk scaled Gaussian unitary ensemble ---
 see e.g.~\cite[Ch.~7]{F} for a precise meaning --- in terms of a particular $\sigma$-form Painlev\'e V transcendent.
 Among many subsequent approaches to problems of this type, Forrester and Witte \cite{FW-2001}, \cite{FW-2002}, \cite{FW-2004}
 proceeded by making direct use of Okamoto's Hamiltonian formulation of the Painlev\'e equations PII--PVI. Of particular
 interest is the characterization obtained for the generalized gap probability
 \begin{equation}\label{def:LUE-ave}
 E_n(s;\alpha,\mu)=\frac{1}{C_n}\int_s^{\infty} \cdots \int_s^{\infty}\prod_{1\leqslant j<k \leqslant n}(\lambda_j-\lambda_k)^2 \prod_{l=1}^n\lambda_l^{\alpha}e^{-\lambda_l}(\lambda_l-s)^{\mu} d\lambda_l,
 \end{equation}
 where the normalisation $C_n$ is as in (\ref{LUE-Zn}). When  $\mu=0$, the above quantity is a particular gap probability
 for the Laguerre unitary ensemble ---  explicitly it is the probability that there are no eigenvalues in the interval $(0,s)$ ---
 while for other values of $\mu$ it is a generalisation of this quantity. Note that in the case $\mu=2$, and with $n$ replaced by
 $n-1$, (\ref{def:LUE-ave}) is proportional to the derivative of the probability that there are no eigenvalues in the interval $(0,s)$, and thus
directly gives the probability density function for the smallest eigenvalue.
		
\begin{thm}\label{T3} (Forrester and Witte \cite{FW-2002}.) Define $U_n(t;\alpha, \mu):= t {d \over dt} \log E_n(t;\alpha, \mu) - \mu n$.
	We have that $U_n(t;\alpha, \mu)$ satisfies the $\sigma$-Painlev\'e V differential equation (\ref{PV}) with parameters
\begin{equation}\label{PVp}
\nu_0 = 0, \quad \nu_1 = - \mu, \quad \nu_2 = n + \alpha, \quad \nu_3 = n.
\end{equation}
\end{thm}

Suppose we change variables $\lambda_l \mapsto \lambda_l s$ in (\ref{def:LUE-ave}), and then $\lambda_l \mapsto 1/ \lambda_l$.
This shows
$$
 E_n(s;\alpha,\mu)=\frac{1}{C_n} s^{(\alpha + \mu + n)n}\int_0^{1} \cdots \int_0^{1}\prod_{1\leqslant j<k \leqslant n}(\lambda_j-\lambda_k)^2 \prod_{l=1}^n\lambda_l^{-(\alpha+ \mu + 2n)}e^{-s/\lambda_l}(1 - \lambda_l)^{\mu} d\lambda_l.
$$
Hence, up to the normalisation $C_n$, we can identify $s^{-n(\alpha + \mu + n)} E_n(s;\alpha,\mu)$ in the
case that $\mu = \beta$, $\alpha \mapsto - \alpha - \mu - 2n$, with the multiple integral
\eqref{eq:Hankel-integral} in the case of the weight (\ref{w2}). Indeed we see that Theorems 	\ref{T2} and \ref{T3}
map to each other with this identification of the functions of $t$ and parameters (with regard to the latter, note
that (\ref{PV}) is symmetric in $\{\nu_j\}$ so the labelling does not matter).

There is also a very simple relationship between the multiple integrals
\eqref{eq:Hankel-integral} with the weights (\ref{w0})	and (\ref{w2}). In the latter, we change variables
$\lambda_l \mapsto \lambda_l/\beta$, $l=1,\dots,n$, make the replacement
$t \mapsto t/\beta$, and take $\beta \to \infty$. Up to a scaling, the deformed Jacobi weight (\ref{w2})
then limits to the deformed Laguerre weight (\ref{w0}), and the same holds true for the corresponding
multiple integrals. In keeping with this, one sees that upon making the change of variables
$t \mapsto t/\beta$ in (\ref{J1}), then taking the limit $\beta \to \infty$, the differential equation
reduces to (\ref{H-equation}), which therefore is a particular degeneration of PV to PIII$'$.

Random matrix theory suggests another degeneration, starting now with \eqref{def:LUE-ave}.
Thus it is well known that the eigenvalues of the Laguerre unitary ensemble in the
neighbourhood of the origin form a well defined large $n$ limiting state upon the so-called
hard edge scaling
$\lambda \mapsto {\lambda \over 4 n}$ (the factor of 4 in the denominator is just for
convenience); see e.g.~\cite[\S 7.2.1]{F}.  This suggests that  \eqref{def:LUE-ave},
with an appropriate modification of the normalisation, admits a well defined limit
upon setting $s=\frac{t}{4n}$ and taking $n\to\infty$.
Indeed one sees that the
$\sigma$-Painlev\'e V differential equation (\ref{PV}) with parameters
(\ref{PVp}) reduces to the
$\sigma$-Painlev\'e III$'$ differential equation (\ref{eq: sigma-equation-1})
with parameters
\begin{equation}\label{def:parameter}
(v_1,v_2)=(\alpha+\mu, \alpha-\mu).
\end{equation}
On the other hand, we know from Remark \ref{R2} that up to a change of scale and linear shift the
$\sigma$-Painlev\'e III$'$ equation is equivalent to (\ref{eq:h}), which for the parameters
(\ref{hv}) has a solution in terms of the generating function $M_n(t)$
(\ref{def: generating function}). To match parameters requires that we set
\begin{equation}\label{hv1}
\alpha = n, \quad \mu = n + \alpha
\end{equation}
in (\ref{def:parameter}); here the variables on the right hand side are
to be considered distinct from the use of $n$ and $\alpha$ in (\ref{def:LUE-ave}).

How then is the hard edge limit of (\ref{def:LUE-ave}) related to $M_n(t)$?
Insight into this question can be obtained by noting from (\ref{hv1}) that the case of
interest in (\ref{def:LUE-ave}) involves $\alpha \in \mathbb Z^+$.
Proceeding independent of the Painlev\'e theory, Forrester and Hughes \cite{FH} showed
that in this circumstance, up to the precise detail of the choice of normalisation
in (\ref{def:LUE-ave}),
\begin{equation}\label{eq: hard edge-limit}
\lim_{n\to\infty}E_n\Big ({t \over 4n};\alpha,\mu \Big ) = e^{-t/4}t^{-\mu\alpha/2}\det \biggl[I_{j-k+\mu}(\sqrt{t} ) \biggr]_{j,k=0}^{\alpha-1}.
\end{equation}
Note that since the entries in the determinant depend of the difference $j-k$, this is a particular
Toeplitz  determinant, whereas according to (\ref{MD}) $M_n(t)$ is a Hankel determinant.
Nonetheless, both are structured determinants, and moreover it is a general fact (see
e.g.~\cite[Eq.~(5.76)]{F}) that with $c_j := {1 \over 2 \pi} \int_{-\pi}^\pi
e^{i j \theta} W(e^{i \theta}) \, d \theta$
\begin{align}\label{Tp}
\det [c_{j-k}]_{j,k=0}^{n-1} & = {1 \over (2 \pi)^n n!} \int_{-\pi}^\pi  \cdots
 \int_{-\pi}^\pi  \prod_{1\le j < k \le n} |
 e^{i \theta_j} - e^{i \theta_k} |^2 \,
 \prod_{l=1}^n  W(e^{i \theta_j}) \, d \theta_l
 \nonumber \\
 & = \Big \langle  \prod_{j=1}^n W(e^{i \theta_j}) \Big \rangle_{\rm CUE},
 \end{align}
 (cf.~(\ref{MT}))
 where the average is over the eigenvalues of a random matrix from the circular unitary
 ensemble (or equivalently a random matrix chosen from the set of unitary matrices
 $U(n)$ with Haar measure; see e.g.~\cite[\S 2.2.]{F}). As
observed in \cite[Eq.~(4.32)]{FW-2002}, applying (\ref{Tp}) to (\ref{eq: hard edge-limit})
tells us that up to normalisation, and assuming $\alpha \in \mathbb Z^+$,
\begin{equation}
\lim_{n\to\infty}E_n\Big ({t \over 4n};\alpha,\mu \Big ) = e^{-t/4} t^{-\mu \alpha/2}
\Big \langle e^{{1 \over 2} \sqrt{t} {\rm Tr} \,(U + \bar{U})} (\det U)^{-\mu}
\Big \rangle_{U \in {\rm CUE}}.
\end{equation}

In \cite{FW-2002}, a direct derivation of the fact that
the RHS of (\ref{eq: hard edge-limit})
relates to a so-called $\tau$-function in
Okamoto's theory of the Hamiltonian theory of PIII$'$ was given; for the latter see the brief summary
in subsection \ref{S2.1} below. The main objective of the present study is to show how this
working naturally implies a second $\tau$-function sequence
(the sequence is labelled by $n$), corresponding to the same
parameters, which can be identified with $M_n(t)$. The required working is carried out in Section \ref{S2} below.

\subsection{Asymptotics of $M_n(t)$}
From both a random matrix theory viewpoint, and for application to the statistics of the Wigner delay
time, it is of interest to quantify the large $n$ form of $M_n(t)$. In the Laguerre unitary ensemble with
the scaling of the eigenvalues
\begin{equation}\label{sL}
\lambda \mapsto n \lambda
\end{equation}
(this has the effect of the leading order density now being on a finite interval), and for a smooth linear
statistic, the mean is of order $n$, the variance is of order unity, and the cumulant of order
$k$ ($k > 2$) decays like $n^{2-k}$. This can be seen from a loop equation analysis; see
e.g.~\cite{FRW}. The fact that the higher order cumulants tend to zero implies that in this setting the limiting distribution is a
Gaussian. We know from \cite[\S 2.1]{CMOS19} that without scaling of the
eigenvalues $\langle L \rangle = {n \over \alpha}$. In fact for the singular linear statistic $L$, without scaling all cumulants of ${1 \over n}L$ are also
of order unity. Their precise form will be discussed in Section \ref{S3}, and related to asymptotic results contained
in \cite{DXZ}. The situation is different in the case $\alpha = n$,
as is relevant to the Wigner time delay problem: the cumulants
then exhibit the behaviour typical of a smooth linear statistic
\cite{MS}. We will use Theorem \ref{T1} to compute their leading
large $n$ form, reclaiming results from \cite{MS}.

\section{Hankel determinants of Bessel functions as a $\tau$-function sequence in Okamoto's theory of PIII$'$}\label{S2}

\subsection{The Hankel determinant}
Our first aim is to place the Hankel determinant corresponding to $M_n(t)$ in the context of 
Okamoto's Hamiltonian formulation of PIII$'$.
Fundamental for this purpose is the observation that
like in
(\ref{eq: hard edge-limit}), the entries are in fact particular Bessel functions, although now of the second kind.

\begin{lem}
With the Hankel determinant defined in \eqref{def: Hankel det}, we have
\begin{equation} \label{Hankel- D-Bes}
D_n(t)=(-1)^{\frac{n(n-1)}{2}}2^{n}t^{\frac{n(n+\alpha)}2}\det\biggl[ K_{j-k+n+\alpha}(2\sqrt{t}) \biggr]_{j,k=0}^{n-1},
\end{equation}
where $K_\nu$ is the modified Bessel function of the second kind.
\end{lem}
\begin{proof}
Upon recalling the integral representation of the modified Bessel function of the second kind
\begin{equation}\label{modified Bes-integral}
K_v(z)=K_{-v}(z)=\frac{1}{2}\left(\frac{z}{2}\right)^{v}\int_0^{\infty} e^{-x-\frac{z^2}{4x}}\frac{dx}{x^{v+1}}, \qquad |\arg z|< \frac{\pi}{4};
\end{equation}
see e.g.~\cite{as,dlmf}, it follows that the moments \eqref{def: moment} can be expressed as
$$
\mu_j(t) = 2 \, t^{\frac{j+\alpha+1}{2}}K_{j+1+\alpha}(2\sqrt{t}).
$$
Substituting in the Hankel determinant and extracting factors gives
$$
D_n(t)
=2^{n}t^{\frac{n(n+\alpha)}2}\det \biggl[K_{j+k+1+\alpha}(2\sqrt{t}) \biggr]_{j,k=0}^{n-1}.
$$
The form \eqref{Hankel- D-Bes}  follows by arranging the columns in reverse i.e.,   $(1,2,...,n)\mapsto (n,n-1,...,1)$.
\end{proof}

As will be detailed below, the significance of this form rests with knowledge of the fact that
$K_{v}(2\sqrt{t})$ satisfies a second order linear differential equation, which determines a seed solution
to a sequence of $\tau$-functions for Painlev\'e III$'$ in Okamoto's theory.

\subsection{The Painlev\'e III$'$ Hamiltonian and $\tau$-function}\label{S2.1}

In Okamoto's Hamiltonian formulation of the Painlev\'e III$'$ equation
\cite{OK}, the Hamiltonian is
\begin{equation}\label{def: H}
tH=p^2q^2-(q^2+v_1q-t)p+\frac{v_1+v_2}{2}q;
\end{equation}
see also \cite[(32.6.24)]{dlmf}. Thus, eliminating $p$ from the Hamilton equations
\begin{equation}\label{Hamiltonian equation}
q'=\frac{\partial H}{\partial p}, \quad p'=-\frac{\partial H}{\partial q},
\end{equation}
gives the Painlev\'e III$'$  equation with parameters $(-4v_2, 4(v_1+1), 4,-4)$,
\begin{equation}\label{PIII equation}
q''=\frac{1}{q}q'^2-\frac{1}{t}q'+\frac{q^2}{t^2}(q-v_2)-\frac{1}{q}+\frac{v_1+1}{t};
\end{equation}
see \cite[(32.6.27)]{dlmf}.

As already remarked, it can be deduced that the modified Hamiltonian
(\ref{hH}) satisfies the nonlinear equation (\ref{eq:h}), and that the rescaling
and linear shift of the modified Hamiltonian \eqref{hH1} satisfies the $\sigma$-form of the
Painlev\'e III$'$ equation
(\ref{eq: sigma-equation-1}).

It is fundamental that  the Hamiltonians associated with the Painlev\'e equations satisfy B\"acklund transformations. For the Painlev\'e III$'$ Hamiltonian \eqref{def: H}, as observed by Okamoto, the B\"acklund transformations can be constructed in terms of elementary operators $s_0$,
$s_1$, and $s_2$. These operators have action on the parameters $v_1$ and $v_2$, and the functions $p$ and $q$, as  given in the following table.
\begin{table}[h]
  \centering
  \begin{tabular}{|c||c|c|c|c|c|}
   \hline
          & $v_1$ & $v_2$ & $p$ & $q$ & $t$  \\ \hline  & & & & & \\
     $s_0$ & $-1-v_2$ & $-1-v_1$ & $\displaystyle\frac{q}{t}\left[ q(p-1)-\frac{1}{2}(v_1-v_2) \right]+1$ & $\displaystyle-\frac{t}{q}$ & $t$ \\ [3ex]  $s_1$ & $v_2$ & $v_1$ & $p$ & $\displaystyle q+\frac{v_2-v_1}{2(p-1)}$ & $t$ \\ [2ex]
     $s_2$ & $v_1$ & $-v_2$ & $1-p$ & $-q$ & $-t$ \\ [1ex]
   \hline
  \end{tabular}
\medskip
  \caption{B\"acklund transformations relevant to the PIII$'$ Hamiltonian \eqref{def: H}; cf.~\cite[Table 4.1]{FW-2002}.} \label{piii-table}
\end{table}

Following Okamoto \cite{OK}, introduce the operator
\begin{equation}
  T_1=s_0s_2s_1s_2.
\end{equation}
Then we have the induced B\"acklund transformations
\begin{equation}
  T_1(H)=H|_{(v_1,v_2)\to T_1(v_1,v_2)} \qquad \textrm{with} \quad T_1(v_1,v_2)=(v_1+1,v_2+1).
\end{equation}
as can be checked from \eqref{def: H} and Table \ref{piii-table}.

Generally a $\tau$-function in a Hamiltonian theory of a Painlev\'e system has the defining property
that its logarithmic derivative with respect to $t$ equals $H(t)$.
By iterating the transformation $T_1$,  a sequence of $\tau$-functions for the Painlev\'e III$'$ can be
defined according to
\begin{equation}\label{def: tau and H sequence}
H[n]=\frac{d}{dt}\ln \tau[n] \qquad \textrm{with} \quad H[n]=T_1^nH=H|_{(v_1,v_2)\to(v_1+n,v_2+n)}.
\end{equation}
As shown by Okamoto \cite{OK}, this sequence has the significant property of
satisfying  the Toda lattice equation
\begin{equation}\label{eq: toda equation}
\delta^2\ln \bar{\tau}[n]=\frac{\bar{\tau}[n-1] \bar{\tau}[n+1]}{\bar{\tau}^2[n]}, \quad \delta:=t\frac{d}{dt},
\end{equation}
where
\begin{equation}\label{def: tau-bar}
\bar{\tau}[n](t)=t^{\frac{n^2}{2}}\tau[n]\Big (\frac{t}{4} \Big ).
\end{equation}
Thus knowledge of $\tau[0]$ and $\tau[1]$ determines the full sequence recursively.
Furthermore, if $\tau[0]=1$, it is shown in \cite{OK} that the full sequence has the  double Wronskian Hankel determinant form
\begin{equation}\label{tD}
\bar{\tau}[n](t) = \det \Big [ \delta^{j+k} \bar{\tau}[1](t) \Big ]_{j,k=0}^{n-1}.
\end{equation}

\subsection{Reduction of $\{ \bar{\tau}[n](t)\}$ to a Toeplitz determinant of modified Bessel functions}
A further key fact is that for a choice of parameters permitting $\tau[0]=1$, the next member
in the sequence is given in terms of modified Bessel functions.

\begin{pro} (Okamoto \cite{OK}.)
The special choice of parameters
$v_1=-v_2 $
in \eqref{def: H} is consistent with setting $\bar{\tau}[0]=1$.
For this choice of parameters,
the first member $\bar{\tau}[1]$ in the $\tau$-function sequence
as specified by \eqref{def: tau and H sequence} is given by
\begin{equation}\label{tau-1-seed solution}
t^\frac{v-1}{2}\bar{\tau}[1](t)=\mathcal{L}_v(\sqrt{t}) \qquad \textrm{with } v := v_1.
\end{equation}
Here $\mathcal{L}_v(t)$ is an arbitrary linear combination of the modified Bessel functions
\begin{equation} \label{def: Bes-Lv}
  \mathcal{L}_v(t)=aI_{v}(t)+be^{v\pi i}K_{v}(t),
\end{equation}
where $a$ and $b$ are constants independent of $v$. Moreover, the subsequent members $\bar{\tau}[n]$ are
expressed in terms of the double Wronskian Hankel determinant
\begin{equation}\label{eq: toda equation-solution}
t^{\frac{n(v-1)}{2}}\bar{\tau}[n](t)=\det\biggl[ \delta^{j+k}\mathcal{L}_v( \sqrt{t}) \biggr]_{j,k=0}^{n-1}, \qquad \delta=t\frac{d}{dt}.
\end{equation}
\end{pro}

\begin{proof}
	With $v_1 = -v_2$ as the $n=0$ case, we see from (\ref{def: H}) and the associated Hamilton equations that we are free to take
	$H=p=0$. Recalling (\ref{def: tau and H sequence}) we see that it is indeed valid to choose $\tau[0] = 1$.
	
It has been shown in \cite{OK} (see also \cite[\S 8.2.4]{F}) that the corresponding first member of the modified
$\tau$-function sequence \eqref{def: tau-bar}, upon multiplication by $t^{-1/2}$, is determined as a solution
of the second-order linear differential equation
$$
t u'' + (v+1) u' - {1 \over 4} u = 0.
$$
The two linearly independent solutions of this equation are the modified Bessel functions given in \eqref{def: Bes-Lv},
thus implying  \eqref{tau-1-seed solution}.

We can check that if $\{\bar{\tau}[n] \}_{n=0,1,\dots}$ is a solution of the Toda lattice equation
\eqref{eq: toda equation} with $\bar{\tau}[0]=1$, then $\{t^{n \kappa}\bar{\tau}[n] \}_{n=0,1,\dots}$
is also a solution for any $\kappa$, obtained by replacing $\bar{\tau}[1]$ by $t^\kappa \bar{\tau}[1]$
in \eqref{eq: toda equation-solution}.
Making use of this result with $\kappa = (v-1)/2$ gives \eqref{eq: toda equation-solution}.
\end{proof}

As found earlier \cite{FW-2002} in the case of \eqref{def: Bes-Lv} with $b=0$, using recursive
properties of the Bessel functions one may simplify the determinant in \eqref{eq: toda equation-solution} such that it is independent of the operator $\delta$. In fact the double Wronskian
Hankel form then reduces to a Toeplitz form.
\begin{pro}\label{P2}
We have
\begin{equation}\label{Bes-det}
 \det \biggl[ \delta^{j+k} \mathcal{L}_v (\sqrt{t}) \biggr]_{j,k=0}^{n-1}=\Big ( {t \over 4 } \Big )^{\frac{n(n-1)}{2}} \det\biggl[ \mathcal{L}_{j-k+v}(\sqrt{t}) \biggr]_{j,k=0}^{n-1}.
\end{equation}
\end{pro}
\begin{proof}
The working of \cite[Prop.~4.4]{FW-2002} is to be followed.
This is possible because $I_{v}(t)$ and $e^{v\pi i}K_{v}(t)$ satisfy the same recurrence relations.
Indeed, we have
$$
\mathcal{L}_v'(t)=\mathcal{L}_{v+1}(t)+\frac{v}{t}\mathcal{L}_v(t), \qquad
\mathcal{L}_v'(t) =\mathcal{L}_{v-1}(t)-\frac{v}{t}\mathcal{L}_v(t)
$$
independent of the constants $a$ and $b$ in \eqref{def: Bes-Lv};
see \cite[Eq.(10.29.2)]{dlmf}. The working of  \cite[Prop.~4.4]{FW-2002} only requires the validity of these
 recurrences to deduce \eqref{Bes-det} from \eqref{eq: toda equation-solution}.
\end{proof}

\begin{cor}  \label{lem-last}
The function
\begin{equation}\label{sigma-1}
\hat{\sigma}_n(t;v)=-t\frac{d}{dt}\ln\left(e^{-t/4}t^{\frac{v^2}{2}}\det \Big [ \mathcal{L}_{j-k+v}(\sqrt{t}) \Big ]_{j,k=0}^{n-1}\right)
\end{equation}
satisfies the $\sigma$-form of the Painlev\'e III$'$ equation \eqref{eq: sigma-equation-1}
with parameters
\begin{equation}\label{2.17}
  (v_1,v_2)=(v+n, -v+n).
\end{equation}
\end{cor}
\begin{proof}
Substituting \eqref{eq: toda equation-solution} and \eqref{Bes-det} into \eqref{sigma-1}, and use of the definitions   \eqref{hH},
\eqref{hH1}, \eqref{def: tau and H sequence} and \eqref{def: tau-bar} shows
\begin{equation}\label{eq: sigma-n-sigma-relation}
\hat{\sigma}_n(t)=-\frac{t}{4}H\Big (\frac{t}{4}\Big )+\frac{t}{4}-\frac{v(n+v)}{2}=-\sigma\Big (\frac{t}{4} \Big )+\frac{t}{8}+\frac{n^2-v^2}{4}.
\end{equation}
The stated result now follows from the theory noted below \eqref{hH}.
\end{proof}

By an appropriate choice of $v$, and with $a=0$ in
\eqref{def: Bes-Lv} the determinant in
\eqref{sigma-1} can be identified with
the Hankel determinant $D_n(t)$ as given in \eqref{Hankel- D-Bes}.

\begin{thm} \label{main-thm}
Let the Hankel determinant $D_n(t)$ be defined in \eqref{def: Hankel det}.  We have
\begin{equation}\label{sigma-2}
t\frac{d}{dt}\ln D_{n}(t)=-\hat{\sigma}_n(4t;n+\alpha)+t-\frac{(n+\alpha)\alpha}{2},
\end{equation}
where $\hat{\sigma}_n(t;v)$ satisfies the
$\sigma$-form of the Painlev\'e III$'$ as specified in Corollary \ref{lem-last}, and furthermore
exhibits the asymptotic behavior
\begin{equation}\label{boundary condtion}
\hat{\sigma}_n(t;v)=\frac{t}{4}+\frac{n}{2}\sqrt{t}+\left(\frac{n^2}{4}-\frac{v^2}{2}\right)+O\Big ({ 1 \over \sqrt{t}} \Big ), \quad t\to\infty.
\end{equation}
\end{thm}

\begin{proof}
  Choose $v=n+\alpha$ in \eqref{sigma-1}. Then we have
  \begin{equation}\label{Mu}
      \hat{\sigma}_n(t;n+\alpha)=-t\frac{d}{dt}\ln\left(e^{-t/4}t^{\frac{(n+\alpha)^2}{2}}\det(\mathcal{L}_{j-k+n+\alpha}(\sqrt{t}))_{j,k=0}^{n-1}\right).
  \end{equation}
  Now in \eqref{def: Bes-Lv} set $a=0$. 
  Simple manipulation of the determinant shows
$$
\det(\mathcal{L}_{j-k+n+\alpha}(2\sqrt{t}))_{j,k=0}^{n-1} = b^n e^{(n+\alpha) n \pi i }
\det ( K_{j-k+n+\alpha}(2\sqrt{t}) )_{j,k=0}^{n-1}.
$$
Using the above two formulas, the right hand side of (\ref{sigma-2}) can be simplified to read
$$
t\frac{d}{dt} \ln\left(t^{\frac{n(n+\alpha)}{2}} \det\biggl[ K_{j-k+n+\alpha}(2\sqrt{t}) \biggr]_{j,k=0}^{n-1} \right ).
$$
Substituting \eqref{Hankel- D-Bes} in the left hand side of \eqref{sigma-2}, and simplifying by noting the factors independent of $t$ do
not contribute  to the logarithmic derivative, precisely the same expression results.

To obtain the asymptotic behavior, we use the integral representation for the Hankel determinant
 \begin{equation*}\label{3.23}
     D_n(t)={t^{\frac{n(n+\alpha)}{2}} \over n!}\int_0^{\infty}...\int_0^{\infty}\prod_{1\leqslant j<k \leqslant n}(\lambda_j-\lambda_k)^2\prod_{l=1}^n\lambda_l^{\alpha}e^{-\sqrt{t}(\lambda_l+\frac{1}{\lambda_l})}d\lambda_l.
  \end{equation*}
 Note that the function $h(\lambda)=\lambda+\frac{1}{\lambda}$, $\lambda>0$ in the exponent has a single minimum  at $\lambda=1$ with
 $h(1)=2$, $h'(1)=0$ and $h''(1)=2$.
 Hence, as $t \to \infty$, the major contribution to the asymptotics of the integral arises from the neighbourhood of $\lambda_k = 1$, $k=1,2,\dots,n$.
Thus we Taylor expand each  $h(\lambda_k)$ in the exponent to second order and so replace each by $2 + (\lambda_k - 1)^2$. Now changing variables by setting
$x_k =  t^{\frac{1}{4}} (\lambda_k - 1)$ gives the asymptotics as $t \to \infty$
  \begin{equation*}\label{Hankel-asy}
     D_n(t)=ce^{-2n\sqrt{t}}t^{\frac{n(n+\alpha)}{2}-\frac{n^2}{4}}
     \Big (1+O\Big ( { 1 \over \sqrt{t}}  \Big ) \Big ),
  \end{equation*}
  for some constant $c$. Substituting this into \eqref{sigma-2} gives \eqref{boundary condtion}.
\end{proof}

\begin{cor}\label{Cf} Set
\begin{equation}\label{ys}
  y_n(t)=t\frac{d}{dt}\ln D_{n}(t)=-\hat{\sigma}_n(4t)+t-\frac{(n+\alpha)\alpha}{2}.
\end{equation}
We have that $y_n(t)$ satisfies the differential equation
\eqref{H-equation}.
	\end{cor}

\begin{proof} We know from Theorem \ref{main-thm} the particular
	differential equation satisfied by $\hat{\sigma}_n(t)$.
	Making the scaling $t \mapsto 4t$ and the linear shift
	as required by the final expression in (\ref{ys}) gives
	\eqref{H-equation}.
		\end{proof}
	
	Since from (\ref{MD}), $M_n(t)$ is proportional
	to $D_n(t)$, we can substitute $M_n(t)$ for
	$D_n(t)$ in Corollary \ref{Cf}. Doing this we reclaim Theorem \ref{T1}, now derived entirely within the
	setting of the Okamoto $\tau$-function theory for
	Painlev\'e III$'$.
	
	\subsection{Relationship to a discrete Painlev\'e equation}
	Discrete Painlev\'e equations are a class of nonlinear difference equations associated with translations on
	crystallographic lattices \cite{Jo19}. B\"acklund transformations have interpretation as such translations,
	and indeed it is well known that $\tau$-function sequences
	constructed in this way satisfy discrete Painlev\'e
	type recurrences.
	In particular, from \cite{FW03} we know
	that the $\tau$-function sequence corresponding to the
	$I$-Bessel determinant in (\ref{eq: hard edge-limit})
	satisfies recurrences giving rise to the
	alternate discrete Painlev\'e II
	equation \cite{RGMM}. In fact the working leading to this
	result proceeds entirely from the algebraic aspects of
	the Okamoto $\tau$-function theory for
	Painlev\'e III$'$, and so holds equally well for $I$
	replaced by the general linear combination
	(\ref{def: Bes-Lv}).
	
	The specific $\tau$-function sequence to be considered is
	\begin{equation}\label{2.x}
	\hat{\tau}[n](t) = \det \Big [ \mathcal{L}_{j-k+v}(\sqrt{t})
	\Big ]_{j,k=0}^{n-1},
	\end{equation}
    where $\mathcal{L}_v(t)$ is a linear combination of the modified Bessel functions given in \eqref{def: Bes-Lv}.
    We know that $\{ \bar \tau[n]\}_{n=0,1,\dots}$ is a $\tau$-function sequence for PIII$'$ with $t$ replaced
    by $t/4$ and parameters (\ref{2.17}) in the Hamiltonian (\ref{def: H}), i.e.~taking the logarithmic
    derivative of $\bar \tau[n]$ gives $H$ as specified.
	Noting
	(\ref{eq: toda equation-solution}), Proposition \ref{P2}
	and (\ref{def: tau-bar}), we see that $t {d \over dt}
	\log \hat \tau[n](t)$ is equal to
	$t {d \over dt}
	\log \bar \tau[n](t)$ plus a constant, and hence
$\{ \hat{\tau}[n]\}_{n=0,1,\dots}$ is also a
$\tau$-function sequence for this same $H$,
but now shifted by a constant on the RHS of
(\ref{def: H}). Specializing to $a=1$, $b=0$ in
(\ref{def: Bes-Lv}), the B\"acklund transformations
	for this sequence were shown in \cite{FW03}
	to imply a coupled recurrence scheme determining
	each $\hat{\tau}[n]$. As already
	commented, the working makes use only of algebraic properties of the transformations,
	and so applies equally to the case of the general linear
	combination in (\ref{def: Bes-Lv}); the only change required
	to the result of \cite{FW03}	is to adjust the initial conditions.
		
		\begin{thm} (Forrester and Witte \cite[Prop.~2]{FW03}.)
	Specify $\hat{\tau}[n]$ as in (\ref{2.x}). Let
	$p_n, q_n$ denote the $p,q$ variables in the Hamiltonian
	(\ref{def: H}) with parameters (\ref{2.17}).
	The sequences $\{ \hat{\tau}[n] \}_{n=0,1,\dots}$,
	$\{p_n\}_{n=0,1,\dots}$, $\{q_n\}_{n=0,1,\dots}$
	satisfy the coupled recurrences
	\begin{align}
	{  \hat{\tau}[n+1]  \hat{\tau}[n - 1] \over
	( \hat{\tau}[n])^2 } \biggr|_{t \mapsto 4t} & = p_n \qquad (n=1,2,\dots) \\
p_{n+1} & = {q_n^2 \over t} (p_n - 1) - {v q_n \over t} + 1
\quad (n=0,1,\dots) \label{pn+1,pn,qn} \\
q_{n+1} & = - {t \over q_n} + {(1+n) t \over q_n ( q_n (p_n - 1) - v) + t } \quad (n=0,1,\dots)
\end{align}
subject to the initial conditions
\begin{align*}
p_0 = 0, \qquad & q_0 = t {d \over dt} \log t^{-v/2}
\mathcal L_v(2\sqrt{t}) , \\
 \hat{\tau}[0] =1, \qquad &  \hat{\tau}[1]|_{t \mapsto 4t} = {\mathcal L}_v(2\sqrt{t}).
\end{align*}
\end{thm}

It is noted in \cite[Proof of Prop.~4.6]{FW-2002} that combining
the recurrences for $\{p_n\}$, $\{q_n\}$ gives
\begin{equation}\label{pqA}
q_{n+1} + {t \over q_n} = {1+n \over p_{n+1}}.
\end{equation}
It is further noted that in addition to the forward equations
relating $p_{n+1}$ and $q_{n+1}$ to $p_n, q_n$, there are also
backward equations, relating $p_{n-1}$ and $q_{n-1}$ to
$p_n, q_n$. For $q_{n-1}$, this reads
$$
q_{n-1} = {t \over  \displaystyle {n \over p_n} - q_n }.
$$
Using the above two formulas to eliminate $p_{n+1}$ and $p_n$ in \eqref{pn+1,pn,qn} then gives
\cite[Prop.~4.6]{FW-2002}
\begin{equation}
{1 + n \over q_n q_{n+1} + t} +
{n \over q_n q_{n-1} + t} = {1 \over q_n} - {q_n \over t} +
{n - v \over t}, \qquad (n=0,1,\dots).
\end{equation}
It is this nonlinear difference equation which is of the type
referred to as the alternate discrete Painlev\'e II
equation \cite{RGMM}.

	\section{Large $n$ asymptotics of $M_n(t)$}\label{S3}
	We know from (\ref{def: generating function})
and (\ref{MT}) that $M_n(t)$, specified	as a ratio of Hankel
determinants in (\ref{MD}), has the interpretation as the exponential moment
generating function for the linear statistic (\ref{L}).
Specifically, denoting the moments by $\{ m_p \}_{p=1}^\infty$
we have
$$
M_n(t) = 1 + \sum_{p=1}^\infty {(-1)^p t^p \over p!} m_p.
$$

From a statistical viewpoint, the corresponding cumulants
$\{ \kappa_p \}_{p=1}^\infty$ can often be of more direct
relevance. In general, while $\kappa_1 = m_1$, one has
$\kappa_2 = m_2 - m_1^2$ which is the variance,
$\kappa_3 = m_3 - 3 m_2 m_1 + 2 m_1^3$ which together
with the variance is used to define the skewness $\gamma$
according to $\gamma = \kappa_3/\kappa_2^{3/2}$ etc. For
a Gaussian distribution, $\kappa_p = 0$ for $p > 2$.
The exponential generating function of the cumulants is obtained from $M_n(t)$
according to
\begin{equation}\label{kM}
\log M_n(t) = \sum_{p=1}^\infty {(-1)^p t^p \over p!} \kappa_p.
\end{equation}
Note that it is $\log M_n(t)$ which is directly characterized in
Theorem \ref{T1}. We can readily compute
\begin{equation}\label{km1}
	\kappa_1 = {n \over \alpha}, \qquad
	\kappa_2 = {n^2 + n \alpha \over  \alpha^2 ( \alpha^2 - 1)},
	\end{equation}
	by seeking a power series solution of the differential
	equation in Theorem \ref{T1} about the origin.
	Note that the first of these requires $\alpha > 0$ to be positive and thus well defined, while the second requires $\alpha > 1$. This is in keeping with the power term in the weight
	(\ref{w0}) being $x^\alpha$, while the singularity at the origin induced by forming the $k$-th moment of (\ref{L})
	is to leading order proportional to $x^{-k}$; for the product to be integrable
	requires $\alpha > k - 1$.
Our interest in this section is in the scaled large $n$ form of
$\{\kappa_p\}$, both for fixed $\alpha$, and the choice $\alpha =
n$ as is  relevant to the Wigner time delay problem.

For fixed $\alpha$, as is consistent with (\ref{km1}),
$\kappa_p/n^p$, $p=1,2,\dots$ tends to a well defined limit.
This can be read off from an earlier asymptotic result
for $M_n(t)$
of Xu et
al.~\cite[Th.~1]{XDZ}, giving the large $n$ form uniformly valid
for any $0< t \le d$ with $d$ fixed, which we state with $t$ replaced by $t/n$ and the limit $n \to \infty$ taken.

\begin{thm}\label{T8}(Xu et al.~\cite{XDZ}.) Define $r(s)$ as the solution analytic
	at the origin of the third order nonlinear equation
	\begin{equation}\label{rs1}
	2 s^2 r' r''' - s^2 (r'')^2 + 2 s r' r'' -
	4 s (r')^3 + \Big ( 2r - {1 \over 4} \Big ) (r')^2 + 1 = 0,
	\end{equation}
	subject to the initial condition $r(0) = {1 \over 8} (1 -
	4 \alpha^2)$. One has
	\begin{equation}\label{rs2}
	\lim_{n \to \infty} M_n \Big ( {t \over n} \Big ) =
	\exp \int_0^t {1 - 4 \alpha^2 - 8 r(2\xi) \over 16 \xi}
\, d \xi.
\end{equation}
	\end{thm}

A characterization of this limit can also be obtained directly
from Theorem \ref{T1}. First note that according to
(\ref{kM}), the quantity $y_n(t)$ in Theorem \ref{T1} has the power series
expansion
\begin{equation}\label{3.3}
y_n(t) = \sum_{p=1}^\infty {(-1)^p  t^p \over (p-1)!}  \kappa_p.
\end{equation}
We know from \cite{XDZ} that for $n \to \infty$,
$y_n(t/n)$ tends to a well defined
limit, $Y(t)$ say, analytic about the origin.
From the viewpoint of the consistency of this power series with
(\ref{3.3}), we must have that $Y(t)$ satisfies the differential
equation which results from
\eqref{H-equation} upon the change of variable $t \mapsto t/n$,
and then equating terms of leading order in $n$ (these occur at order $n^2$). This equation reads
\begin{equation}\label{Y}
(tY'')^2 = 1 + \alpha^2 (Y')^2 + 2 \alpha Y' - 4 (t Y' - Y)
(Y')^2,
\end{equation}
which can be checked to admit a unique power series solution,
and so as an alternative characterization to (\ref{rs2}) we have
\begin{equation}\label{rs3}
\lim_{n \to \infty} M_n \Big ( {t \over n} \Big ) =
\exp \int_0^t {Y (\xi) \over  \xi}
\, d \xi.
\end{equation}
In fact it is already known from Xu et al.~\cite[Eq.~(1.30)]{XDZ2015} and \cite[Eq.~(2.12)]{XDZ} that
the
function $r(s)$ as defined in Theorem \ref{T8} also satisfies
\begin{equation}\label{rs4}
s^2 (r'')^2 - 2s (r')^3 + {8 r - 1 \over 4} (r')^2 + 2 \alpha r'
-1 = 0.
\end{equation}
Replacing $s$ by $2s$, then substituting $r(2s) = - 2 Y(s) +
(1 - 4 \alpha^2)/8$, reduces (\ref{rs4}) to (\ref{Y}).

In addition to the characterisation of the exponential generating
function for the singular statistic (\ref{L}), in the case of the Laguerre unitary ensemble, through the $\sigma$-Painlev\'e
III$'$ equation (\ref{H-equation}), we have also revised in Theorem
\ref{T2} an analogous characterisation of the exponential
generating function for (\ref{L}), now in the
Jacobi unitary ensemble. The Jacobi weight has the same power
singularity $x^\alpha$ for $x \to 0^+$ as the Laguerre weight.
Expecting the linear statistic (\ref{L}) to probe this hard
edge region, we would therefore anticipate that in the limit
$n \to \infty$, upon appropriate scaling, the quantity
$D_n(t)$ in Theorem \ref{T2} tends to the right hand side of
(\ref{rs3}). Indeed scaling $t$ in (\ref{J1}) by $t \mapsto
t/n^2$ we reclaim (\ref{Y}). Such a limiting relation has also recently been identified by  Chen et al. \cite[Thm. 7.]{CCF2019}.

We now turn our attention to the case $\alpha = n$. By an
analysis of the large $n$ form of the recurrence satisfied
by $\{ \kappa_p\}$ as implied by the nonlinear equation
\eqref{H-equation}, Mezzadri and Simm \cite{MS} established
that $\{ n^{2p-2} \kappa_p |_{\alpha = n}\}$ has a well defined
limit, and moreover gave a specification in terms of a generating
function.

\begin{thm} (Mezzadri and Simm \cite{MS}.)
	Specify $F(t)$ as the power series solution of the first order
	nonlinear equation
	\begin{equation}\label{FF}
	2 F(t) + F'(t) - 4 t F'(t) - 6t (F'(t))^2 + 4 F(t) F'(t) = 1.
	\end{equation}
	One has
	\begin{equation}\label{FF1}
	\sum_{p=1}^\infty  \lim_{n \to \infty} n^{2p-2}
	\kappa_p \Big |_{\alpha = n} {t^p \over (p-1)!} = F(t).
	\end{equation}
\end{thm}

In keeping with our discussion below Theorem \ref{T8}, it should
be possible to identify $F(t)$ as being related to a limiting
solution of \eqref{H-equation}. To be consistent with the
established fact that $n^{2p - 2} \kappa_p |_{\alpha = n}$ as
a well defined large $n$ limit, and recalling too the
relation (\ref{3.3}) between $y_n(t)$ and $\{\kappa_p\}$, we
see that after setting $\alpha = n$ we should change variables $t \mapsto -n^2 t$, and
furthermore replace $y_n(-n^2 t)$ by $n^2 F(t)$. Doing this, and
equating the leading order term  shows that in addition to
(\ref{FF}), $F$
satisfies
\begin{equation}\label{FF2}
(1 - F')^2 - 4 F' (F' + 1) (t F' - F) = 0.
\end{equation}

The compatibility of (\ref{FF2}) and (\ref{FF}) can be
checked directly. We begin by considering $F$ to be specified
by (\ref{FF}). Then we can check by using this to
substitute for $t(F')^2$ in the factor $-4t (F')^3$ that (\ref{FF2}) is valid if and only if
\begin{equation}\label{FF3}
 (4 - 8F)F'  - (F')^2 - 4 F (F')^2 + 4 t (F')^2 = 3.
\end{equation}
In particular, the validity or otherwise of (\ref{FF3}) is unchanged by adding  it
to (\ref{FF2}). This leaves the equation obtained by multiplying
${2 \over 3}F'$ times (\ref{FF}), and is thus valid by our
assumption that (\ref{FF}) is valid.

\section*{Acknowledgements}
The work of DD was supported by grants from the City University of Hong Kong (Project No. 7005252), and grants from the Research Grants Council of the Hong Kong Special Administrative Region, China (Project No. CityU 11303016, CityU 11300520).
The work of PJF was supported by the Australian Research Council (ARC) through the ARC Centre of Excellence for Mathematical and Statistical frontiers (ACEMS) and the ARC grants DP170102028 and DP210102887, and further acknowledges the support of DD in hosting a visit to City University during September 2019 when this work was completed. The work of SXX was supported by National Natural Science Foundation of China under grant numbers 11971492, 11571376 and 11201493.

\end{document}